\documentclass[runningheads,11pt]{llncs}
\usepackage[margin=1.2in]{geometry}

\usepackage{mathtools, amssymb}
    \def\le{\leqslant}
    \def\ge{\geqslant}
    \def\BB{\mathbf B}
    \def\RR{\mathbf R}
    \allowdisplaybreaks
    \def\CCC{\mathcal C}
    \def\LLL{\mathcal L}
    \def\OOO{\mathcal O}
    \def\VS{{V \setminus S}}
    \def\poly{\textsf{poly}}
    \def\EE{\mathop{\mathbf E}}

\usepackage{tikz-cd, booktabs}

\usepackage[colorlinks, allcolors=purple!50!blue!90!black]{hyperref}

\begin{document}

                                   \title
          {Semirandom Planted Clique via 1-norm Isometry Property}
                                      
                                   \author
                  {Venkatesan Guruswami \and Hsin-Po Wang} 
                                      
                                 \institute
     {University of California, Berkeley\and National Taiwan University}
                                      
                                 \maketitle
                                      
\begin{abstract}
    We give a polynomial-time algorithm that finds a planted clique of size
    $k \ge \sqrt{n \log n}$ in the semirandom model, improving the
    state-of-the-art $\sqrt{n} (\log n)^2$ bound.  This \emph{semirandom
    planted clique problem} concerns finding the planted subset $S$ of $k$
    vertices of a graph $G$ on $V$, where the induced subgraph $G[S]$ is
    complete, the cut edges in $G[S; V \setminus S]$ are random, and the
    remaining edges in $G[V \setminus S]$ are adversarial.

    \smallskip

    An elegant greedy algorithm by B{\l}asiok, Buhai, Kothari, and Steurer
    \cite{BBK24} finds $S$ by sampling inner products of the columns of the
    adjacency matrix of $G$, and checking if they deviate significantly from
    typical inner products of random vectors.  Their analysis uses a suitably
    random matrix that, with high probability, satisfies a certain restricted
    isometry property.  Inspired by Wootters's work on list decoding, we put
    forth and implement the $1$-norm analog of this argument, and
    quantitatively improve their analysis to work all the way up to the
    conjectured optimal $\sqrt{n \log n}$ bound on clique size, answering one
    of the main open questions posed in \cite{BBK24}.
\end{abstract}

\section{Introduction}

    An important and well-studied class of computational problems concerns
    finding an atypical structure amidst ambient noise.  One such
    quintessential challenge is the fundamental \emph{planted clique} problem
    from combinatorial optimization, where we are given a graph $G$ on $n$
    vertices that contains a clique of size $k$ (structure), while the
    remaining edges are randomly placed (noise).  See
    Figure~\ref{fig:question} for an example.

    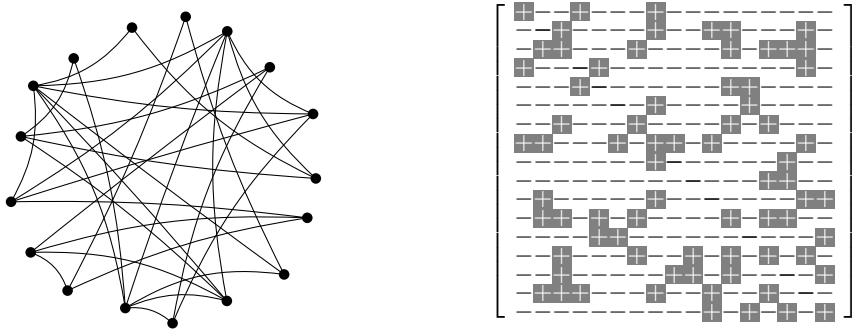
\begin{figure}
        \centering
        \def\n{17}
        \def\p{0.2}
        \begin{tikzpicture}
            \pgfmathsetseed{0988365024}
            \foreach \v in {1, ..., \n}{
                \fill (360/\n*\v : 2) + (360*rnd : 0.1)
                    circle (2pt) coordinate (v\v);
            }
            \foreach \u in {1, ..., \n}{
                \edef\uu{\the\numexpr \u + 1}
                \foreach \v in {\uu, ..., \n}{
                    \pgfmathtruncatemacro\b{rnd < \p}
                    \ifnum \b = 1
                        \draw (v\u) to [bend left=(\u+\n/2-\v+rand)*3] (v\v);
                    \fi
                }
            }
            
            \foreach \u in {3, 7, 12, 14}{
                \foreach \v in {3, 7, 12, 14}{
                    \ifnum \u < \v
                        \draw (v\u) to [bend left=(\u+\n/2-\v+rand)*3] (v\v);
                    \fi
                }
            }
        \end{tikzpicture}
        \hfil
        \begin{tikzpicture} [yscale=-1]
            \small
            \tikzset{
                0/.style={node contents={$-$}},
                1/.style={
                    node contents={$+$},
                    fill=gray, text=white, inner sep=0pt, inner ysep=0.5pt
                },
            }
            \pgfmathsetseed{0988365024}
            \foreach \u in {1, ..., \n}{
                \foreach \v in {\u, ..., \n}{
                    \pgfmathtruncatemacro\b{rnd < \p}
                    \draw [scale=0.25]
                        (\u, \v) node [\b] {} (\v, \u) node [\b] {};
                }
            }
            \foreach \u in {3, 7, 12, 14}{
                \foreach \v in {3, 7, 12, 14}{
                    \draw [scale=0.25] (\u, \v) node [1] {};
                }
            }
            \draw [scale=0.25] (\n/2, \n/2) + (0.5, 0.4) node
                {$\left[ \rule{4.5cm}{0cm} \rule{0cm}{2.2cm} \right]$};
        \end{tikzpicture}
        \caption{
            An instance of planted clique problem.  Left is the graph
            and right is the adjacency matrix.  Answer is in
            Figure~\ref{fig:answer}.
        }                                                \label{fig:question}
    \end{figure}

    A situation in which the clique is obvious is when $k \gg \sqrt{n \log
    n}$.  In this situation, since the degree of every vertex in a uniformly
    random graph is in the range $n/2 \pm \OOO(\sqrt{n \log n})$ with high
    probability, one can find the planted $k$-clique by simply picking out
    the $k$ vertices with the highest degrees.  A classic spectral algorithm
    improves this bound and succeeds in finding the planted clique when $k
    \ge \Omega(\sqrt{n})$ \cite{AKS98}.

    However, if $k$ is much smaller than $\sqrt n$ but still\footnote{ A
    random graph will have cliques of size $2 \log_2 n$ with high probability
    \cite[Section~4.5]{AlS16}, so we must assume greater $k$ for the planted
    clique to ``stand out''.} much greater than $\log n$, these simple
    algorithms cease to work.  In fact, no efficient algorithm is known in
    this regime, and the well-known planted clique hypothesis asserts that no
    such algorithm exists. In other words, there seems to be a \emph{
    statistical vs.\ computational gap} of $\log n$ vs.\ $\sqrt n$, in which
    range certain tasks are information-theoretically possible but
    (seemingly) computationally intractable.  The study of this phenomenon,
    and which computational problems with statistical inputs exhibit such
    gaps, is a central theme in average-case complexity that has received
    accelerated attention in recent years.

    In addition to the large $k$ needed for the degree-based algorithm to
    work, the algorithm is also fragile and easy to trick.  If we add $k$ new
    edges to a non-clique vertex, the top $k$ high-degree vertices will
    contain a false positive.  The algorithm is thus confounded by even
    modest fluctuations from a fully random base graph. To study the
    robustness of clique-finding algorithms, Feige and Kilian introduced the
    insightful and influential \emph{semirandom} model \cite{FeK01}.  In one
    version of this model, $S$ is the vertices of the planted clique, edges
    between $S$ and $\VS$ (where $V$ is the set of vertices) are placed
    randomly, and the edges within $\VS$ can be configured adversarially with
    access to the random edges.

    Unlike the random model, for which simple solutions exist for $k \ge
    \Omega(\sqrt{n})$, in the semirandom model it is not obvious whether we
    can identify a clique as large as, say, $n^{2/3}$. In fact, the goal
    itself has to be relaxed from finding the specific planted clique $S$,
    which is impossible because the adversary can add up to $(n-k)/k$ cliques
    of size $k$ in $G[\VS]$.  One can, instead, ask the algorithm to find a
    list of at most $(1 + o(1)) n/k$ cliques that includes the planted clique
    $S$.  Steinhardt \cite{Ste18} showed that when $k \ll \sqrt{n}$, this
    task of finding a list of $(1 + o(1)) n/k$ cliques including the correct
    clique is impossible, \emph{even information-theoretically}. This shows a
    sharp contrast between the random vs.\ semirandom planted clique problem,
    as the random case is statistically solvable for $k$ as small as
    $\OOO(\log n)$.

    \begin{figure}
        \centering
        \def\n{17}
        \def\p{0.2}
        \begin{tikzpicture}
            \pgfmathsetseed{0988365024}
            \foreach \v in {1, ..., \n}{
                \fill (360/\n*\v : 2) + (360*rnd : 0.1)
                    circle (2pt) coordinate (v\v);
            }
            \foreach \u in {1, ..., \n}{
                \edef\uu{\the\numexpr \u + 1}
                \foreach \v in {\uu, ..., \n}{
                    \pgfmathtruncatemacro\b{rnd < \p}
                    \ifnum \b = 1
                        \draw (v\u) to [bend left=(\u+\n/2-\v+rand)*3] (v\v);
                    \fi
                }
            }
            \foreach \u in {3, 7, 12, 14}{
                \foreach \v in {3, 7, 12, 14}{
                    \ifnum \u < \v
                        \draw [red, line width=0.8pt, line cap=round]
                            (v\u) to [bend left=(\u+\n/2-\v+rand)*3] (v\v);
                    \fi
                }
            }
        \end{tikzpicture}
        \hfil
        \begin{tikzpicture} [yscale=-1]
            \small
            \tikzset{
                0/.style={node contents={$-$}},
                1/.style={
                    node contents={$+$},
                    fill=gray, text=white, inner sep=0pt, inner ysep=0.5pt
                },
            }
            \pgfmathsetseed{0988365024}
            \foreach \u in {1, ..., \n}{
                \foreach \v in {\u, ..., \n}{
                    \pgfmathtruncatemacro\b{rnd < \p}
                    \draw [scale=0.25]
                        (\u, \v) node [\b] {} (\v, \u) node [\b] {};
                }
            }
            \foreach \u in {3, 7, 12, 14}{
                \foreach \v in {3, 7, 12, 14}{
                    \draw [scale=0.25] (\u, \v) node [1, fill=red] {};
                }
            }
            \draw [scale=0.25] (\n/2, \n/2) + (0.5, 0.4) node
                {$\left[ \rule{4.5cm}{0cm} \rule{0cm}{2.2cm} \right]$};
        \end{tikzpicture}
        \caption{
            The planted clique is highlighted.
        }                                                  \label{fig:answer}
    \end{figure}
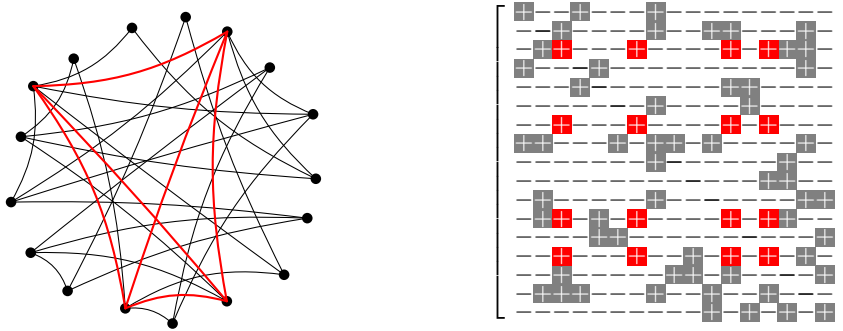

    This raises the natural challenge of whether we can solve semirandom
    planted clique for $k$ approaching the Steinhardt threshold of
    $\sqrt{n}$.  This question has received a fair bit of attention in recent
    years. Algorithms based on (different) semidefinite programming (SDP)
    relaxations were given to solve this problem when $k \gtrsim n^{2/3}
    \log^{1/3} n$ in \cite{CSV17}, and when $k \gtrsim n^{2/3}$ in
    \cite{MMT20}. The latter work also suggested that their SDP should likely
    fail when $k \ll n^{2/3}$. More recently, Buhai, Kothari, and Steurer
    \cite{BKS23} gave a complicated algorithm based on sum-of-squares SDP to
    solve the problem all the way to $k \ge n^{1/2+\varepsilon}$; the runtime
    of the algorithm is $n^{\OOO(1/\varepsilon)}$ and explodes as $k$
    approached Steinhardt's threshold of $n^{1/2}$.

    Finally, B{\l}asiok, Buhai, Kothari, and Steurer \cite{BBK24} gave a
    delightfully elegant greedy algorithm for the problem that works as long
    as $k \gtrsim \sqrt{n} (\log n)^2$. It was also shown in these works
    \cite{BKS23,BBK24} that when $k \gtrsim \sqrt{n \log n}$, one can prune
    any polynomial-sized list containing the planted set $S$ to a list of
    size $(1+o(1)) n/k$. 
    
    Given this backdrop, a natural question is whether one can solve
    semirandom planted clique, i.e., output a polynomial-sized list including
    the planted set $S$, for $k$ all the way down to $\OOO(\sqrt{n \log n})$,
    which is believed to be the right threshold.
    
    This is precisely what we achieve in this short paper.

    \begin{theorem} [main]                                   \label{thm:main}
        The semi-random planted clique problem on $n$ vertices, where the
        planted clique has size $k \ge \OOO(\sqrt{n \log n})$ and random
        edges have density $1/2$, has a $\poly(n)$-time algorithm that
        outputs a list of $(1 + o(1)) n/k$ guesses of the clique, and is
        correct with probability $0.99$.
    \end{theorem}

    In fact, our algorithm is the same as that of \cite{BBK24}.  We provide a
    simpler analysis that leads to the quantitatively stronger bound.  The
    starting point of the algorithm of \cite{BBK24} is the following base
    strategy.
    \begin{itemize} \itemsep=0.5ex
        \item Consider $A$ the $\pm 1$ adjacency matrix of $G$ (with $+1$
            denoting presence of an edge).  Let $A^v$ be the $v$th column of
            $A$ for $v \in V$.
        \item Randomly choose a vertex $u \in V$.
        \item Let $S_u \subset V$ be the set of vertices $v \in V$ such that
            the inner product between $A^v$ and $A^u$ is in the range $k \pm
            \OOO(\sqrt n)$.  This $S_u$ will be one guess of $S$.
        \item Repeat the two steps above $\Theta(n/k)$ times so that with
            high probability, we are lucky at least once and $u$ is in $S$.
            It remains to show that when $u$ is in $S$, the guess $S_u$
            differ from $S$ by at most $o(k)$ vertices.
            (The extra vertices $S_u \setminus S$ can easily be pruned out;
            the missing vertices $S \setminus S_u$ can easily be found.)
    \end{itemize}
    This strategy can handle the region $k \ge \tilde\OOO(n^{3/4})$
    \cite[Section~2.1]{BBK24}.
    
    As it turns out, the obstacle is that there are only $n$ choices of
    $u \in V$, so it is relatively easy\footnote{ To clarify, it is not
    proved that the adversary has a way to fool the algorithm when $k <
    \tilde\OOO(n^{3/4})$; it is just that the analysis given in
    \cite[Section~2.1]{BBK24} needs $k \ge \tilde\OOO(n^{3/4})$.} for
    the adversary to manipulate inner products.  B{\l}asiok \textit{et al.}
    strengthened this strategy in \cite{BBK24} by comparing more inner
    products.
    \begin{itemize} \itemsep=0.5ex
        \item Randomly choose \emph{three} vertices $u_1, u_2, u_3 \in V$.
        \item Let $S_{u_1u_2u_3} \subset V$ be the set of vertices $v
            \in V$ such that the inner product between $A^v$ and $A^{u_1}
            \odot A^{u_2} \odot A^{u_3}$ (where $\odot$ denotes entry-wise
            product of vectors)  ``looks right'' (this will be made precise
            later). This $S_{u_1u_2u_3}$ will be one guess for $S$.
        \item Repeat the two steps above $\Theta((n/k)^3)$ times so that
            with high probability, we are lucky at least once and $u_1$,
            $u_2$, $u_3$ all belong to $S$.  It remains to show that when
            $u_1$, $u_2$, $u_3$ are all in $S$, the guess $S_{u_1u_2u_3}$
            differ from $S$ by $o(k)$ vertices.
    \end{itemize}
    The improved algorithm works in the region $k \ge \OOO(\sqrt{n} (\log n
    )^2)$.  The improvement comes from the fact that it is harder for the
    adversary to manipulate inner products as triples $(u_1, u_2, u_3) \in
    V^3$ are sampled from a much larger sample space.\footnote{A natural
    question is what the bounds on $k$ are if we sample pairs $(u_1, u_2)$,
    quadruples $(u_1, u_2, u_3, u_4)$, and so on.  That is, why stop at
    triples?  A rule of thumb is one can work with $k > n^a$, where $a$ is
    $3/(3 + \text{the length of tuples})$.  But we also know that $\sqrt n$
    cannot be beaten.  Hence, stopping at triples is the best option.  See
    \cite[Section~2.2]{BBK24} for a discussion of sampling pairs.}

    The crux of the analysis of \cite{BBK24} is to control deviations of
    inner products via the restricted isometry property (RIP) of a certain
    random matrix. This RIP, which concerns $\ell^2$-norms, is hard to
    control. Our key insight is to realize that one can work with a simpler
    $\ell^1$ version of the RIP, and to control it with a simpler and
    quantitatively tighter analysis. This aspect is directly inspired by a
    similar use of $\ell^1$-RIP in the novel analysis of list-decoding
    properties of random linear codes by Wootters \cite{Woo13}. 

\paragraph{The paper is organized as follows.}

    In Section~\ref{sec:inner}, we review the inner-product algorithms
    in greater detail.
    In Section~\ref{sec:isometry}, we explain how their performance
    analysis is related to whether certain random matrices have RIP.
    In Section~\ref{sec:1norm}, we discuss list decoding and how proof
    techniques from there can be applied here.
    In Section~\ref{sec:wrap}, we wrap up the proof of the main theorem.

\section{Check Inner Products}                              \label{sec:inner}

    In this section, we review \cite{BBK24}'s greedy algorithms that are
    based on sampling inner products between the columns of the adjacency
    matrix of $G$.
    
    Since it is easier to talk about deviations when the expectation is zero,
    we assume that the adjacency matrix $A \in \{+1, -1\}^{n\times n}$
    records an edge as $+1$ and a non-edge as $-1$.  Let $u$ and $v$ be two
    vertices.  Consider the inner product between $A^u$ and $A^v$, the $u$th
    and the $v$th columns of $A$.  There are three cases.
    \begin{itemize} \itemsep=0.5ex
        \item If both $u$ and $v$ are in the clique, then the inner product
            looks like
            \[ k + (\text{sum of $n - k$ random signs}), \]
            which is about $k \pm \OOO(\sqrt{n-k})$ by the central limit
            theorem, or $k \pm \sqrt n$ for short.
        \item If $u$ is in the clique but $v$ is not, then the inner product
            looks like the sum of $n$ random signs, which is about $\pm \sqrt
            n$.
        \item If neither of them is in the clique, then the inner product is
            the sum of $n$ adversarial signs, which can be anything.
    \end{itemize}
    Accordingly, if $u$ happens to be in the clique, we can separate $v \in
    S$ from $v \notin S$ by checking if the inner product is $k \pm \sqrt n$
    or $\pm \sqrt n$.  We can, for instance, define $S_u \subset V$ to be
    those $v \in V$ such that the inner product is at least $k/2$.

    Now for the improved algorithm that samples three vertices $u_1, u_2, u_3
    \in V$ instead of one, we first compute the entry-wise product of
    $A^{u_1}$, $A^{u_2}$, and $A^{u_3}$.  That is, we let $T^{u_1u_2u_3} \in
    \{+1, -1\}^{n\times1}$ be the column vector whose $w$th row is the
    product $A^{u_1}_w \cdot A^{u_2}_w \cdot A^{u_3}_w \in \{+1, -1\}$.  Then
    we can compute the inner products of $T^{u_1u_2u_3}$ with $A^v$ for all
    vertices $v \in V$.

    Observe that, when $u_1, u_2, u_3$ are in the clique, the $w$th entry of
    $T^{u_1u_2u_3}$ is
    \begin{equation*}
        T^{u_1u_2u_3}_w =
        \begin{cases*}
            1 & if $w \in S$, \\
            \text{product of three random signs} & if $w \notin S$.
        \end{cases*}
    \end{equation*}
    So $T^{u_1u_2u_3}$ plays the same role as $A^u$ in the earlier algorithm
    in the sense that the inner product between $T^{u_1u_2u_3}$ and $A^v$
    would be $k \pm \sqrt n$ if $v \in S$, and $\sqrt n$ if $v \notin S$.
    This is what ``looks right'' means in the introduction.  We define
    $S_{u_1u_2u_3} \subset V$ to be those $v \in V$ such that the inner
    product is at least $k/2$.

    We claim that all we want now
    is that $S_u$ is very close to $S$ for a considerable fraction of $u \in
    S$.  And for the latter algorithm, we want that $S_{u_1u_2u_3}$ is very
    close to $S$ for a considerable fraction of triples $(u_1, u_2, u_3) \in
    S^3$.  We do not need all $u$ (and all triples) to be good because we can
    simply sample more times.  We do not need $S_u$ (and $S_{u_1u_2u_3}$) to
    be exactly $S$ because we can post-process a close enough $S_u$ (and
    $S_{u_1u_2u_3}$) to get $S$, as the lemma below justifies.

    \begin{lemma} [Claim~5.9 of \cite{BKS23}]                \label{lem:post}
        Fix a small constant $\gamma > 0$.  With probability at least $0.99$
        the following happens: For all $S' \subset V$ such that $|S \cap S'|
        \ge (1 - \gamma) k$ and $|S'| \le (1 + \gamma) k$, every $v \in S$
        is connected to $(1 - \gamma) k - 1$ or more vertices in $S'$.
        Meanwhile, every $v \notin S$ is connected to $(1/2 + 2\gamma) k +
        \OOO(\sqrt{k \log n})$ or fewer vertices in $S'$.
    \end{lemma}

    \begin{proof}
        Every $v \in S \setminus S'$ connects to all $(1 - \gamma) k$
        vertices in $S \cap S' \subset S'$; Every $v \in S \cap S'$ connects
        to one less because it does not connect to itself; this proves the
        first statement.  For the second statement, note that every $v \notin
        S$ connects to $S$ via random edges, and $\OOO(\sqrt{k \log n})$ is
        the high-probability tail bound.
    \end{proof}

\section{The Reduction to Restricted Isometry Property}  \label{sec:isometry}

    In this section, we explain how \cite{BBK24} reduces the performance of
    the inner-product--based algorithm to whether a certain random matrix has
    a form of restricted isometry property (RIP) with high probability.  A
    matrix $M \in \mathbb{R}^{m\times n}$ is said to have the
    $(t,\varepsilon)$-RIP if for all $t$-sparse vectors $z \in \mathbb{R}^n$,
    $(1-\varepsilon) \|z\|_2 \le \|M z\|_2 \le (1+\varepsilon) \|z\|_2$.
    This property is useful to recover an unknown sparse ``signal" $z$ using
    only $m \ll n$ ``measurements" $M z$.

    The readers will see (a) the randomness of the ``measurement'' matrix $M$
    corresponds to the random edges between $S$ and $\VS$, (b) the $z$ in the
    RIP statement
    \begin{equation}
        \|Mz\|_2 \approx (1 \pm \varepsilon) \|z\|_2,          \label{eq:rip}
    \end{equation}
    corresponds to the set of triples $(u_1, u_2, u_3) \in S^3$ where
    $S_{u_1u_2u_3}$ misclassifies some vertex $v \notin S$, and (c) the
    sparsity condition on $x$ corresponds to the success rate of the
    inner-product--based algorithm.

\subsection{The Simple Version}

    In order to explain the reduction better, let us first focus on the
    simpler version (the one that samples $u$, not triples).  From the
    perspective of the adversary, there are two options that can fool the
    inner-product algorithm.
    \begin{itemize} \itemsep=0.5ex
        \item The adversary can manipulate edges such that, when $u$ is not
            in the clique, $S_u$ appears to be a clique different from $S$.
        \item The adversary can manipulate edges such that, when $u$ is in
            the clique, a $v \notin S$ appears to be in $S$.
    \end{itemize}
    The first case does not pose any threat as the algorithm is allowed to
    output a list of candidate cliques.  There can be many wrong cliques and
    we only need to make sure that the correct clique is in.
    
    For the second option to be a threat, the adversary should realize that
    it is not enough to fabricate a bad inner product between only one pair
    of vertices.  This is because if a guess $S_u$ intersects $S$ at $k \pm
    o(k)$ vertices, one can post-process $S_u$ to get to $S$ per se
    (Lemma~\ref{lem:post}).  As a result, the adversary should seek to
    manipulate the edges within $\VS$ to maximize the number of pairs $(u, v)
    \in S \times (\VS)$ such that $v$ appears to be in $S$ when we compute
    the inner product between $u$ and $v$.

    Now, suppose that we can show the following statement.
    \begin{quote}
        \color{red!80!black}
        $(\heartsuit)$ For any fixed $v \in \VS$, the number of $u \in S$
        leading to bad inner products, i.e., at least $k/2$, is less than
        $\OOO(n^2/k^2)$.
    \end{quote}
    Then by a double counting argument we can see that (a) there are
    $\OOO(n^3/k^2)$ bad pairs $(u, v) \in S \times (\VS)$ and (b) for more
    than 50\% of $u \in S$, the number of bad $v \in \VS$ is less than
    $\OOO(n^3/k^3)$.  Equate $\OOO(n^3/k^3)$ to $o(k)$, we infer that the
    algorithm should work when $k > \omega(n^{3/4})$.

    It remains to reduce $(\heartsuit)$ to an RIP inequality of the form
    \eqref{eq:rip}.  For a fixed $v \in \VS$, let $B \subset S$ be the
    collection of bad $u$'s, that is, $\langle A^u, A^v \rangle \ge k/2$ and
    we want to show that $|B| \le \OOO(n^2/k^2)$.  Summing the inner
    products over $B$, we obtain
    \[  
        |B| \frac k2
        \le \sum_{u \in B} \langle A^u, A^v \rangle
        = \Bigl\langle \sum_{u \in B} A^u, A^v \Bigr\rangle
        = \Bigl\langle \sum_{u \in B} A^u_S, A^v_S \Bigr\rangle
        + \Bigl\langle \sum_{u \in B} A^u_\VS, A^v_\VS \Bigr\rangle.
    \]
    Here, $A^u_S$ is the sub-vector of $A^u$ where we keep the rows
    corresponding to vertices in $S$.  The same goes for $A^u_\VS$, $A^v_S$,
    and $A^v_\VS$.
    
    Note that $\bigl\langle \sum_{u \in B} A^u_S, A^v_S \bigr\rangle$ is the
    boring part as the randomness in $A^v_S$ makes the expectation zero.  The
    other term, the interesting part, can be bounded with the Cauchy--Schwarz
    inequality
    \[
        \Bigl\langle \sum_{u \in B} A^u_\VS, A^v_\VS \Bigr\rangle
        \le \Bigl\| \sum_{u \in B} A^u_\VS \Bigr\|_2
        \cdot \bigl\| A^v_\VS \bigr\|_2
    \]
    All entries are $\pm1$, so $\|A^v_\VS\|_2$ is just $\sqrt{n - k}$.  The
    other $2$-norm can be understood as
    \[
        \Bigl\| \sum_{u \in B} A^u_\VS \Bigr\|_2
        = \bigl\| A^S_\VS 1_B \bigr\|_2.
    \]
    $A^S_\VS \in \{+1, -1\}^{(n-k)\times k}$ is the sub-matrix of $A$ where
    we only keep the rows corresponding to $\VS$ and the columns
    corresponding to $S$; it plays the role of the measurement matrix $M$.
    On the other hand, $1_B$ is the indicator vector of $B \subset S$, it
    plays the role of the $x$ with sparsity $|B|$.  The argument so far is
    summarized by the following proposition.

    \begin{proposition} [Section~2.1 of \cite{BBK24}]
        Suppose that $A^S_\VS / \sqrt{n - k}$ is shown to have
        $(\OOO(n^2/k^2), \OOO(1))$-RIP.  Then $(\heartsuit)$ holds and the
        simple inner-product algorithm works for $k \ge \tilde\OOO(n^{3/4})$.
    \end{proposition}

    However, one should not expect any meaningful bound from the RIP
    literature because, as much as $A^S_\VS$ is random, it is not a wide
    matrix but a tall one.  In \cite{BBK24}, $\|A^S_\VS 1_B\|_2$ is bounded
    using the fact that such a random matrix likely has an operator norm
    $\OOO(\sqrt n)$, so $\|A^S_\VS 1_B\|_2 \le \OOO(\sqrt{n} \sqrt{|B|})$.
    Overall we have $|B|k/2 \le \OOO(\sqrt{n} \sqrt{|B|} \sqrt{n})$, which
    leads to $|B| \le \OOO(n^2/k^2)$, as $(\heartsuit)$ needed.

\subsection{The Triple Version}

    RIP comes into play when we consider the triple version and the statement
    \begin{quote}
        \color{red!80!black}
        $(\diamondsuit)$ for any fixed $v \in \VS$, the number of $(u_1, u_2,
        u_3) \in S^3$ leading to bad inner products, i.e., at least $k/2$, is
        less than $\OOO(n^2/k^2)$.
    \end{quote}
    Then by a double counting argument we can see that (a) there are
    $\OOO(n^3/k^2)$ bad quadruples $(u_1, u_2, u_3, v) \in S^3 \times (\VS)$
    and (b) for more than 50\% of triples in $S^3$, the number of bad $v \in
    \VS$ is less than $\OOO(n^3/k^5)$.  Equate $\OOO(n^3/k^5)$ to $o(k)$, we
    infer that the algorithm should work when $k > \omega(n^{1/2})$.

    Fix a $v \in \VS$.  Let $\BB$ be the collection of bad triples $(u_1,
    u_2, u_3) \in S^3$ that misclassifies $v$.  Let $1_\BB \in \{0,
    1\}^{k^3\times1}$ be the indicator vector of $\BB$.  Let $T \in \{+1,
    -1\}^{n\times k^3}$ be the matrix whose $(u_1, u_2, u_3)$th column is
    $T^{u_1u_2u_3}$, the entry-wise product of $A^{u_1}$, $A^{u_2}$, and
    $A^{u_3}$.  To prove $(\diamondsuit)$, we follow the same logic as above:
    summing inner products, linearity, splitting rows, and Cauchy--Schwarz.
    \begin{align}
        |\BB| \cdot \frac k2
        & \le \sum_{{(u_1,u_2,u_3)\in \BB}}
            \Bigl\langle T^{u_1u_2u_3}, A^v \Bigr\rangle
        = \Bigl\langle \sum_{{(u_1,u_2,u_3)\in \BB}}
            T^{u_1u_2u_3}, A^v \Bigr\rangle                            \notag
        \\ & = \bigl\langle T 1_\BB, A^v \bigr\rangle
        = \bigl\langle T_\VS 1_\BB, A^v_\VS \bigr\rangle
            + \bigl\langle T_S 1_\BB, A^v_S \bigr\rangle               \notag
        \\ & \le \bigl\| T_\VS 1_\BB \bigr\|_2
            \cdot \bigl\| A^v_\VS \bigr\|_2
            + (\text{zero-mean boring part})                    \label{eq:cs}
        \\ & = \bigl\| T_\VS 1_\BB \bigr\|_2 \cdot \sqrt{n-k}
            + (\text{zero-mean boring part}),                          \notag
    \end{align}
    Now that $T_\VS$ is a wide random matrix, we can expect
    \[
        \bigl\| T_\VS 1_\BB \bigr\|_2
        \le (1 + 1) \sqrt{n - k} \bigl\| 1_\BB \bigr\|_2
    \]
    as $T_\VS / \sqrt{n - k}$ likely\footnote{ Note that the entries of $T$
    are not independent as they are algorithmically generated from a random
    source $A$.  That being said, all that is needed for the RIP to hold with
    high probability is that $\EE[T^\top T]$ is a multiple of the identity
    matrix, where the expectation is taking over the randomness of $A$.} has
    RIP with error term $\varepsilon \coloneqq 1$ and sparsity $|\BB| =
    \OOO(n^2/k^2)$.  This subsection is thus summarized by the following
    proposition.

    \begin{proposition} [also Section~2.2 of \cite{BBK24}]
        Assume that $T_\VS / \sqrt{n - k}$ is shown to have $(\OOO(n^2/k^2),
        1)$-RIP.  Then $(\diamondsuit)$ holds and the triple-sampling
        algorithm works for $k \ge \tilde\OOO(n^{1/2})$.
    \end{proposition}

    It is not hard to see that the bound on $k$ is tight if the
    Cauchy--Schwarz inequality \eqref{eq:cs}
    \[
        \bigl\langle T_\VS 1_\BB, A^v_\VS \bigr\rangle
        \le \bigl\| T_\VS 1_\BB \bigr\|_2
        \cdot \bigl\| A^v_\VS \bigr\|_2
    \]
    assumes equality.  The equality is achieved when $T_\VS 1_\BB$ is
    parallel to $A^v_\VS$, a $\pm1$-vector.  But it is unlikely that all the
    entries of $T_\VS 1_\BB$ assume similar absolute values.  In the next
    section, we replace Cauchy--Schwarz with
    \begin{equation}
        \bigl\langle T_\VS 1_\BB, A^v_\VS \bigr\rangle
        \le \bigl\| T_\VS 1_\BB \bigr\|_1
        \cdot \bigl\| A^v_\VS \bigr\|_\infty.              \label{eq:1-infty}
    \end{equation}
    The equality is achieved if the adversary constructs $A^v_\VS$ by
    matching the signs of the entries of $T_\VS 1_\BB$.  This is always
    achievable.  So hopefully we can get a better bound on $k$.  That said,
    it is nontrivial to bound $\|T_\VS 1_\BB\|_1$ out of thin air.  We are
    lucky that a similar bound has been established in the list decoding
    literature, which is the topic of the next section.

\section{List Decodability and the 1-norm Analog of RIP}    \label{sec:1norm}

    In this section, we show how to bound the $1$-norm $\|T_\VS 1_\BB\|_1$
    in \eqref{eq:1-infty}.  The key is to show an inequality of the form
    \[
        \bigl\| T_\VS 1_\BB \bigr\|_1 \le
        \OOO(n \sqrt{|\BB|}) + (\text{deviation term}),
    \]
    where $\sqrt{|\BB|}$ reflects the norm of $1_\BB$ and the deviation term
    should be small.

    Before we proceed to bounding $\|T_\VS 1_\BB\|_1$, we want to briefly
    elaborate on the context of Wootters's work \cite{Woo13}.  Our proof of
    Theorem~\ref{thm:main} is almost a replica of hers because it is driven
    by the serendipitous discovery that \cite{Woo13} and \cite{BBK24} face
    similar challenges.

\subsection{List Decodability}

    In coding theory, there is a codebook $\CCC \subset \{+1, -1\}^n$ and we
    want that, whenever we observe a message $y \in \{+1, -1\}^n$, we should
    be able to decode the original message $u \in \CCC$ by minimizing the
    Hamming distance to the observation
    \[
        u \stackrel?=
        \operatorname*{argmin}_{x\in\CCC} \text{HDist}(x, y).
    \]
    But sometimes the $x$ that minimizes the distance is not the correct
    message.  So in the theory of list coding, we give attention to a list of
    candidates
    \[ \LLL \coloneqq \text{HBall}_r(y) \cap \CCC, \]
    where $\text{HBall}_r(y)$ is a Hamming ball centered at $y$ and $r$ is
    the \emph{search radius} (which is an estimate of how distorted $y$ can
    be).  We hope that the list size $|\LLL|$ is small so that, given other
    clues (such as side channels), we can recover the correct message from
    $\LLL$.
    In particular, we hope that $|\LLL| < |\CCC|$ for otherwise we are not
    simplifying the problem at all.
    This leads to the question of how to construct a codebook $\CCC$
    with specified alphabet size, block length, code rate, search radius, and
    list size.

    As it turns out, determining the region of these five parameters where
    list-decodable codes exist is a difficult problem.  One of the many
    difficulties is that the list size, $|\LLL|$ maximized over all $y$, is
    hard to control.  In \cite{Woo13}, Wootters made the following
    observation\footnote{ Here we focus on binary alphabet $\{+1, -1\}$,
    while \cite{Woo13} covers larger alphabets.}
    \begin{gather*}
        \text{two codewords $x$ and $y$ are $r$-close}
        \\ \Updownarrow
        \\ \text{inner product between $x$ and $y$ is $\ge n - 2r$}
    \end{gather*}
    and thus
    \begin{gather*}
        \text{there is a large list $\LLL$ of $x$'s
            that are $r$-close to $y$}
        \\ \Updownarrow
        \\ \text{inner product between every $x \in \LLL$ and $y$ is large}
        \\ \Downarrow
        \\ \langle \text{the sum of the $x$'s in $\LLL$}, y \rangle
            \text{ is large}
        \\ \Updownarrow
        \\ \text{$\|$the sum of the $x$'s in $\LLL\|_1$ is large}
        \\ \Updownarrow
        \\ \text{$\|\Phi \cdot 1_\LLL\|_1$ is large
            (so we need to bound $\|\Phi \cdot 1_\LLL\|_1$)}
    \end{gather*}
    where $\Phi \in \{+1, -1\}^{n\times|\CCC|}$ is a matrix whose columns are
    codewords of $\CCC$.
    
    One can see that both $\|T_\VS 1_\BB\|_1$ and $\|\Phi 1_\LLL\|_1$ are the
    $\ell^1$-norm of a matrix-vector product, where the matrix is random.
    Moreover, there is a deeper parallel: While $T_\VS$ is algorithmically
    generated from $A$, Wootters' matrix $\Phi$ is also algorithmically
    generated from the generator matrix of a random linear code $\CCC$.

\subsection{The $1$-norm Analog of RIP}

    Now that the readers can see the parallel between $T_\VS 1_\BB$ and $\Phi
    1_\LLL$, we move on to the actual bound.  To begin, we center the random
    variable $\|T_\VS 1_\BB\|_1$ by subtracting its expectation
    \begin{align}
        \max_\BB \Bigl\{ \bigl\|T_\VS 1_\BB \bigr\|_1 \Bigr\}
        & \le \max_\BB \Bigl\{\EE_A \bigl[
            \bigl\| T_\VS 1_\BB \bigr\|_1 \bigr] \Bigr\}     \label{eq:mean1}
        \\ & \mkern2mu + \max_\BB \Bigl\{
            \bigl\| T_\VS 1_\BB \bigr\|_1
            - \EE_A \bigl[ \bigl\| T_\VS 1_\BB \bigr\|_1 \bigr]
            \Bigr\}.                                          \label{eq:sym1}
    \end{align}
    It suffices to upper bound the two terms on the right-hand side.

    For the first term $\max_\BB \{ \EE_A[\|T_\VS 1_\BB\|_1] \}$, note that
    the first row of $T_\VS 1_\BB$ is a sum of $|\BB|$ random signs.  If the
    entries of $T_\VS$ are iid, the sum will be roughly $\pm \sqrt{|\BB|}$.
    But they are not; they are only pairwise independent.  Let $X_1, X_2,
    \dotsc, X_{|\BB|}$ be the first row of $T_\VS 1_\BB$, Jensen's inequality
    gives
    \[
        \EE_A [|X_1 + X_2 + \dotsb + X_{|\BB|}|]^2
        \le \EE_A [(X_1 + X_2 + \dotsb + X_{|\BB|})^2]
        = \EE_A [X_1^2 + X_2^2 + \dotsb + X_{|\BB|}^2] = |\BB|.
    \]
    So the same bound $\sqrt{|\BB|}$ applies.  There are $|\VS| = n - k$
    rows, so 
    \begin{equation}
        \EE_A \bigl[ \bigl\| T_\VS 1_\BB \bigr\|_1 \bigr]
        \le (n - k) \cdot \sqrt{|\BB|} \le n \sqrt{|\BB|}    \label{eq:mean2}
    \end{equation}
    gives an upper bound on \eqref{eq:mean1}.
    
    To bound \eqref{eq:sym1}, we want to apply Markov's inequality to obtain
    an upper bound that is valid 99\% of time.  Therefore, we take the
    expectation over $A$ and attempt to bound
    \[
        \EE_A \Bigl[
            \max_\BB \Bigl\{ \|t\|_1 - \EE_A [\|t\|_1] \Bigr\}
        \Bigr],
        \qquad t \coloneqq T_\VS 1_\BB.
    \]
    A standard symmetrization argument upper bounds this expectation with
    \begin{equation}
        \EE_A \Bigl[ \max_\BB \Bigl\{\|t\|_1
            - \EE_A [\|t\|_1] \Bigr\} \Bigr]
        \le \OOO\Bigl( \EE_A \EE_g \Bigl[ \max_\BB
            \{ \langle g, t \rangle \} \Bigr] \Bigr) \label{eq:sym2}
    \end{equation}
    where $g$ is a vector of $n - k$ iid standard Gaussian random variables.
    See \cite[(7)]{Woo13} for the corresponding part in Wootters's work.  See
    Appendix~\ref{app:sym} for details of the symmetrization argument.

    Now the right-hand side of \eqref{eq:sym2} involves the maximum of some
    linear combinations of Gaussian random variables.  Those are Gaussian by
    themselves, so the following bound applies:
    \begin{equation}
        \EE [\max\{\text{Gaussian rvs}\}]
        \le \text{constant} \cdot \sqrt{\text{variance}
        \cdot \log(\text{\# Gaussian rvs})}.                  \label{eq:sym3}
    \end{equation}
    This bound is a consequence of the union bound, which does not care if
    the Gaussian rvs are correlated.  Clearly the number of Gaussian rvs is
    the number of ways we can choose $\BB$, which is $\binom{n}{|\BB|}$.  But
    there is a better bound.
    
    When controlling highly correlated Gaussian rvs, the sheer number of
    variables is less important than the ``effective dimension'', i.e., the
    number of Gaussian ``seeds'' that generated everything else.  This can be
    seen by the following comparison: $\max_{j\in[m]} \{g_j\}^2$ grows like
    $\log m$ as $m$ increases; but $\max_{j\in[m]} \{\cos(j)g_1 +
    \sin(j)g_2\}^2$ converges to a fixed Chi-squared distribution.
    Therefore, when considering the $\max_\BB$ of $\langle g, t \rangle$, a
    crucial observation is that since $\langle g, t \rangle$ is the sum of
    $\langle g, T^{u_1u_2u_3}_\VS \rangle$ over $(u_1, u_2, u_3) \in \BB$, it
    is a convex combination of $|\BB| \langle g, T^{u_1u_2u_3}_\VS \rangle$.
    It remains to bound the maximum of those inner products.

    \begin{lemma}
        As $g$ varies, we have
        \begin{equation}   
            \max_\BB \langle g, t \rangle
            \le |\BB| \max_{(u_1,u_2,u_3)\in S^3}
            \bigl\langle g, T^{u_1u_2u_3}_\VS \bigr\rangle
            \le \OOO(|\BB| \sqrt{(n - k) \log(k^3)}),         \label{eq:sym4}
        \end{equation}
        where the first inequality always holds and the second is a
        high-probability bound.
    \end{lemma}

    \begin{proof}
        The first inequality, as explained above, is because $t$ is the
        average of $|\BB| T^{u_1u_2u_3}_\VS$ over $(u_1, u_2, u_3) \in \BB$.
        The second inequality is \eqref{eq:sym3} with the variance being $n -
        k$ and the number of Gaussians being $k^3$.
    \end{proof}

    \begin{remark}
        It is temping to go one step further and bound the maximum of
        $\langle g, T^{u_1u_2u_3}_\VS \rangle$ using the fact that they are
        linear combinations of $g_1, g_2, \dotsc, g_{n-k}$.  That is,
        \[
            \max_{(u_1,u_2,u_3)\in S^3} \bigl\{
                \bigl\langle g, T^{u_1u_2u_3}_\VS \bigr\rangle \bigr\}
            \le (n - k) \max_{j\in[n-k]} \{g_j\}
            \le \OOO((n - k) \sqrt{\log(n - k)}).
        \]
        This, however, weakens the bound.  This is because $n - k$ used to be
        in the square root as the variance term, which is now a linear
        rescaling term.  Meanwhile, the number of Gaussians decreases from
        $k^3$ to $n-k$; this improves the bound slightly but not enough to
        compensate for $\sqrt{n - k}$.
    \end{remark}

    With everything above, we can upper bound \eqref{eq:1-infty}.
    \begin{align*}
        \bigl\langle T_\VS 1_\BB, A^v_\VS \bigr\rangle
        & \le \bigl\| T_\VS 1_\BB \bigr\|_1
            \cdot \bigl\| A^v_\VS \bigr\|_\infty
        = \bigl\| T_\VS 1_\BB \bigr\|_1
        = \eqref{eq:mean1} + \eqref{eq:sym1}
        \\ & \le \eqref{eq:mean2} +  \eqref{eq:sym4}
        = \OOO \bigl(n \sqrt{|\BB|} \bigr)
            + \OOO \bigl( |\BB| \sqrt{n \log n} \bigr).
    \end{align*}
    A reminder is that $\eqref{eq:sym1} \le \eqref{eq:sym4}$ is a
    high-probability bound.  We can now improve \eqref{eq:cs} 
    \[
        |\BB| \cdot \frac k2
        \le \bigl\langle T_\VS 1_\BB, A^v_\VS \bigr\rangle
            + \bigl\langle T_S 1_\BB, A^v_S \bigr\rangle
        = \OOO \bigl(n \sqrt{|\BB|} \bigr)
            + \OOO \bigl( |\BB| \sqrt{n \log n} \bigr)
            + \bigl\langle T_S 1_\BB, A^v_S \bigr\rangle
    \]
    Only one term $\bigl\langle T_S 1_\BB, A^v_S \bigr\rangle$ remains to be
    bounded.  Note that it has zero mean and, by Hoeffding's inequality,
    deviates away from zero by $\Omega(|\BB| \sqrt{n \log n})$ with
    probability at most $ 1/n^{10}$ \cite[Lemma~3.2]{BBK24}.  So it
    contributes $\OOO(|\BB| \sqrt{n \log n})$ most of the time.  When $k \ge
    c \sqrt{n \log n}$ for some large enough absolute constant $c$, we can
    conclude from the above that \[ k |\BB| \le \OOO(n \sqrt{|\BB|}) \] and
    thus $|\BB| \le \OOO(n^2/k^2)$, as desired in $(\diamondsuit)$.

\section{Putting Everything Together}                        \label{sec:wrap}

    In this section, we present a formal proof of Theorem~\ref{thm:main}.  We
    begin with the algorithm.
    \begin{itemize} \itemsep=0.5ex
        \item Sample $u_1, u_2, u_3 \in V$ randomly.
        \item Compute $T^{u_1u_2u_3}$ as the entry-wise product of $A^{u_1}$,
            $A^{u_2}$, and $A^{u_3}$.
        \item For each $v \in V$, compute the inner product $\langle
            T^{u_1u_2u_3}, A^v \rangle$.  Let $S_{u_1u_2u_3}$ be those $v \in
            V$ such that the inner product is at least $k/2$.
        \item Post-process $S_{u_1u_2u_3}$ by adding/keeping vertices that
            have many neighbors in $S_{u_1u_2u_3}$ (Lemma~\ref{lem:post}).
        \item Repeat the above steps $\OOO(n^3/k^3)$ times.
        \item Prune the $\OOO(n^3/k^3)$ guesses down to a list of $(1 + o(1))
            n/k$ guesses by removing pairs of cliques that overlap too much
            (\cite[Appendix~B]{BBK24}).
    \end{itemize}
    
   \noindent The validity of the algorithm follows from the argument below:

    \begin{itemize} \itemsep=1ex
        \item Thanks to Lemma~\ref{lem:post}, if any $S_{u_1u_2u_3}$ matches
            the clique $S$ with up to $o(k)$ errors, the post-processing will
            recover $S$ with high probability.
        \item Thanks to \cite[Appendix~B]{BBK24}, the intersection of a fake
            clique $S'$ and the true clique $S$ is usually $o(k)$, so we can
            remove largely-overlapping pair of cliques and will be left with
            a list of mostly-disjoint cliques.  The number of mostly-disjoint
            cliques is at most $(1 + o(1)) n/k$ with high probability.
        \item It remains to show that (a) with high probability, one of our
            choices of $u_1, u_2, u_3$ are all in $S$ and (b) when (a)
            happens, with high probability, $S_{u_1u_2u_3}$ matches $S$ with
            up to $o(k)$ errors.
        \begin{itemize}
            \item Statement (a) is obvious as the probability that $u_1, u_2,
                u_3$ are all in $S$ is $k^3/n^3$, and we try $\OOO(n^3/k^3)$
                times.
            \item Statement (b) can be rephrased as follows: for the majority
                of choices of $(u_1, u_2, u_3) \in S^3$, the number of bad $v
                \in \VS$ is $o(k)$.  By a double counting argument, (b) is a
                consequence of the diamond statement $(\diamondsuit)$, that
                for any fixed $v \in \VS$, the number of bad $(u_1, u_2, u_3)
                \in S^3$ (which is denoted $|\BB|$) is $\OOO(n^2/k^2)$.
            \item Finally, to prove $(\diamondsuit)$ $|\BB| \le
                \OOO(n^2/k^2)$, follow the last paragraph of
                Section~\ref{sec:1norm}.
        \end{itemize}
    \end{itemize}

    This finishes the proof of Theorem~\ref{thm:main} and concludes our
    paper.

\section*{Acknowledgments}

    Supported in part by NSF CCF-2211972 and a Simons Investigator award.

\appendix

\section{The Symmetrization Argument}                         \label{app:sym}

    In this appendix, we explain the symmetrization argument used to proved
    \eqref{eq:sym2}.  The idea is to treat the centering term as iid copies
    that take ``early expectation''.  We then delay the expectation so that
    the new copies are on the same footing as the original random variables.
    This allows us to bound them as twice the same thing.

    A useful trick called \emph{contraction principle} often follows the
    symmetrization argument and control iid sums of bounded random variables.
    The idea is that we do not know, or need, a bound better than Gaussian
    tail.  So we might as well treat them as Gaussians to begin with.

    Let $t$ be $T_\VS 1_\BB$, a random vector in $\RR^{n-k}$.  Recall that
    $T_\VS$ is algorithmically generated by $A$, and $A$ is the ultimate
    source of randomness by the problem description.  The full argument of
    \eqref{eq:sym2} goes as follows:
    \begin{align*}
        \kern2em & \kern-2em
        \EE_A \Bigl[ \max_\BB \Bigl\{
            \|t\|_1 - \EE_A [\|t\|_1] \Bigr\} \Bigr]
        = \EE_A \Bigl[
            \max_\BB \Bigl\{ \|t\|_1 - \EE_{A'} [|t'|_1] \Bigr\} \Bigr]
        = \EE_A \Bigl[
            \max_\BB \Bigl\{ \sum_j |t_j| - \EE_{A'} [|t'_j|] \Bigr\} \Bigr]
        \\ & \le \EE_A \EE_{A'} \Bigl[
            \max_\BB \Bigl\{ \sum_j |t_j| - |t'_j| \Bigr\} \Bigr]
        = \EE_A \EE_{A'} \EE_\varepsilon \Bigl[
            \max_\BB \Bigl\{ \sum_j\varepsilon_j(|t_j|-|t'_j|) \Bigr\} \Bigr]
        \\ & \le \EE_A \EE_{A'} \EE_\varepsilon \Bigl[
            \max_\BB \Bigl\{ \sum_j \varepsilon_j |t_j| \Bigr\} +
            \max_\BB \Bigl\{ \sum_j \varepsilon_j |t'_j| \Bigr\} \Bigr]
        = 2 \EE_A \EE_\varepsilon \Bigl[
            \max_\BB \Bigl\{ \sum_j \varepsilon_j t_j \Bigr\} \Bigr]
        \\ & = \sqrt{2\pi} \EE_A \EE_\varepsilon \Bigl[
            \max_\BB \Bigl\{ \sum_j \EE_g [|g_j|]
            \varepsilon_j t_j \Bigr\} \Bigr]
        \le \sqrt{2\pi} \EE_A \EE_\varepsilon \EE_g \Bigl[
            \max_\BB \Bigl\{ \sum_j |g_j| \varepsilon_j t_j \Bigr\} \Bigr]
        \\ & = \sqrt{2\pi} \EE_A \EE_g \Bigl[
            \max_\BB \Bigl\{ \sum_j g_j t_j \Bigr\} \Bigr]
         = \sqrt{2\pi} \EE_A \EE_g \Bigl[
            \max_\BB \{ \langle g, t \rangle \} \Bigr]
    \end{align*}
    Here, $A'$ and $t'$ are independent copies of $A$ and $t$, and
    $\varepsilon_1, \varepsilon_2, \dotsc, \varepsilon_{n-k}$ are iid
    Rademacher random variables.  The first inequality is by pulling
    $\EE_{A'}$ out of a convex function.  Now that they are on the same
    footing, we can add random signs $\varepsilon$ and change nothing.  The
    second inequality is by the triangle inequality of $\max_\BB$ as a norm.
    This make two copies of $\sum_j \varepsilon_j |t_j|$ so we no longer need
    $t'$.  The third inequality is by pulling $\EE_g$ out of a convex
    function.  This helps us introduce $g$ out of thin air.  As commented
    above, introducing $g$ does not weaken the bound because a sum like
    $\sum_j \varepsilon_j t_j$ is already quite Gaussian; so doing this only
    makes the picture clearer.

\bibliographystyle{alphaurl}
\bibliography{RipClique-42}

\end{document}